\documentclass[journal,12pt,onecolumn,draftclsnofoot,]{IEEEtran}

\usepackage{amsmath}
\usepackage{amssymb}
\usepackage{amsthm}
\usepackage{siunitx} 
\usepackage{graphicx}

\usepackage[usenames,dvipsnames,svgnames,table]{xcolor}
\definecolor{webgreen}{rgb}{0,.5,0}
\usepackage[linktoc=page,linkcolor=RoyalBlue,colorlinks=true,citecolor=webgreen]{hyperref}

\newcommand{\ie}{i.\,e.}

\newcommand{\eg}{e.\,g.}

\theoremstyle{plain}
\newtheorem{lemma}{Lemma}
\newtheorem{theorem}{Theorem}
\newtheorem{corollary}{Corollary}
\newtheorem{definition}{Definition}

\theoremstyle{definition}
\newtheorem{example}{Example}

\def\ed{($\epsilon$,\,$\delta$)}
\def\Prob{\mathbb{P}}
\def\R{\mathbb{R}}
\def\md{\mathbf{d}}
\def\double{{\prime\prime}}

\DeclareMathOperator{\Var}{Var}
\DeclareMathOperator*{\argmin}{arg\,min}

\begin{document}

\title{Optimal Differentially Private Mechanisms for Randomised Response}
\author{Naoise Holohan, Douglas J.\ Leith, Oliver Mason%
\thanks{Naoise Holohan and Douglas J.\ Leith are with the School of Computer Science and Statistics, Trinity College, Dublin 2, Ireland (e-mail: \href{mailto:nholohan@tcd.ie}{nholohan@tcd.ie}; \href{mailto:doug.leith@tcd.ie}{doug.leith@tcd.ie})}%
\thanks{Oliver Mason is with the Department of Mathematics\slash Hamilton Institute, Maynooth University, Co.\ Kildare, Ireland \& Lero, the Irish Software Research Centre (e-mail: \href{mailto:oliver.mason@nuim.ie}{oliver.mason@nuim.ie})}}

\maketitle

\addcontentsline{toc}{section}{Abstract}

\begin{abstract}
We examine a generalised Randomised Response (RR) technique in the context of differential privacy and examine the optimality of such mechanisms. Strict and relaxed differential privacy are considered for binary outputs. By examining the error of a statistical estimator, we present closed solutions for the optimal mechanism(s) in both cases. The optimal mechanism is also given for the specific case of the original RR technique as introduced by Warner in 1965.
\end{abstract}


\section{Introduction}

\subsection{Background}
\IEEEPARstart{S}{tanley}~L.\ Warner first proposed the Randomised Response (RR) technique as a means to eliminate bias in surveying in 1965 \cite{War65}.  Respondents would be handed a spinner by the surveyor to decide which of two questions the respondent would answer, for example,
\begin{enumerate}
\item Have you ever cheated on your spouse\slash partner?
\item Have you always been faithful to your spouse\slash partner?
\end{enumerate}

Respondents would spin the spinner in private and answer the given question truthfully with a `yes' or `no'.  Respondents would be afforded \emph{plausible deniability} as the surveyor would not know the question to which the answer refers.  This would encourage respondents to engage with the survey and answer the question truthfully.  The spinner can be replaced by any appropriate randomisation device, such as coin flips, dice or drawing from a pack of cards.

A rich body of literature now exists on RR. The inefficiencies of Warner's original RR model have been examined by a number of authors and many new RR models have been proposed.  These include the unrelated question model \cite{GAS69}, the forced response model \cite{Bor71}, Moor's procedure \cite{Moo71} and two-stage RR models \cite{MS90, Man94}.  More comprehensive lists of RR models can be found in \cite{Kru13,BIZ15}.

RR is actively used in surveying when asking questions of a sensitive nature.  Examples include surveys on doping and drug use in elite athletes \cite{SUS10}, cognitive-enhancing drug use among university students \cite{DSF15}, faking on a CV \cite{DDH03}, corruption \cite{Gin10}, sexual behaviour \cite{CDJ14}, and child molestation \cite{FL88}. 

Researchers remain divided on the effectiveness of RR.  While some works have shown RR to be an improvement on different survey techniques, including direct questioning (where no randomisation is involved), \cite{HGB00,GG75,LSO04,Kru12,TF81}, others remain sceptical on its advantage \cite{WS94,WP13,LHG97}.  Public trust in RR has also been shown to be lacking \cite{CJ08}.

Separately, differential privacy has emerged as a model of interest in privacy-preserving data publishing since being presented in 2006 \cite{Dwo06}.  Differential privacy gives a quantitative mathematical definition to measure the level of privacy achieved in a given data release.  This definition determines the amount of manipulation that needs to be applied to the data to achieve the desired level of privacy.  Under differential privacy, privacy is quantified by how statistically indistinguishable the privacy-preserved outputs from two similar datasets are.

When applied to randomised response, where the output from a single individual is binary, differential privacy requires the output from any two individuals to be statistically indistinguishable, to a specified degree.  

\subsection{Our Results}
In this paper we examine a generalisation of Warner's original RR technique, and establish conditions under which such a model satisfies differential privacy.  By calculating the estimator of minimal variance, we determine the optimal differentially private RR mechanism.  We examine strict $\epsilon$-differential privacy and relaxed ($\epsilon$,$\delta$)-differential privacy. Complete solutions for the optimal mechanisms are presented for both cases. The optimal mechanism is also given for Warner's RR model satisfying {\ed}-differential privacy.

\subsection{Related Work}
The application of differential privacy to randomised response has been limited to date. \cite{WWH14} examined using randomised response to differentially privately collect data, although their analysis only considered strict $\epsilon$-differential privacy and a comparison of its efficiency with respect to the Laplace mechanism, a mechanism popular in the differential privacy literature.

Randomised response has been used in conjunction with differential privacy in a more general context in the form of \emph{local privacy}, also known as \emph{input perturbation}. For example, extreme mechanisms for local differential privacy have been studied in \cite{KOV14, HLM16b}, while differential privacy was applied to social network data in the form of graphs with randomised response in \cite{KSK14}. Outside randomised response and local privacy, optimal mechanisms in differential privacy have received some attention, including work on strict differential privacy \cite{GV14} and relaxed differential privacy \cite{GV13b}.

\subsection{Structure of Paper}
We begin in Section~\ref{sc:prelim} with an introduction to the Randomised Response (RR) technique, and derive the statistical estimator and associated bias and error; we also present Warner's original RR model.  We introduce differential privacy in Section~\ref{sc:dp} and present a number of preliminary results for later use in Section~\ref{sc:prelimresults}.

The main results are given in Sections~\ref{sc:optsdp}, \ref{sc:optrdp} and \ref{sc:optwar}, relating to strict differential privacy, relaxed differential privacy and Warner's model respectively.  Concluding remarks are given in Section~\ref{sc:conc}.

\section{Randomised Response}\label{sc:prelim}

\subsection{Introduction}

We are looking to determine the proportion $\pi$ of people in the population possessing a particular sensitive attribute, where possession of the attribute is binary. We conduct a survey on $n$ individuals of the population by uniform random sampling with replacement.

A single respondent's answer $X_i\in\{0,1\}$ is a randomised version of their truthful answer $x_i\in\{0,1\}$, in order to protect their privacy.  The randomised response will therefore not definitively reveal a respondent's truthful answer.  By convention, a value of \num{1} denotes possession of the sensitive attribute, while \num{0} denotes that the respondent does not possess the attribute.  We denote by $N$ the number of randomised responses that return \num{1}, hence $N=\sum_{i \in [n]} X_i$ where $[n]=[1,n] \cap \mathbb{Z}$. We are therefore looking to estimate $\pi$ from $\frac{N}{n}$.

\subsection{Generalised RR Model}\label{sc:genrr}

In keeping with standard notation, $(\Omega, \mathcal{F}, \Prob)$ denotes a probability space. $X_i: \Omega \to \{0,1\}$ is then a random variable for each $i \in [n]$, dependent on the truthful value $x_i$.  We define the randomised response mechanism by
\begin{equation}\label{eq:rrmech}
\Prob(X_i=k \mid x_i = j)=p_{jk},
\end{equation}
which leads us to defining the design matrix of the mechanism as follows.

\begin{definition}[Design Matrix]
A randomised response mechanism as defined in (\ref{eq:rrmech}) is uniquely determined by its design matrix,
$$P=\left(\begin{array}{cc} p_{00} & p_{01} \\ p_{10} & p_{11}\end{array}\right).$$
For the probability mass functions of each $X_i$ to sum to \num{1}, we require $p_{00} + p_{01}=1$ and $p_{10} + p_{11} = 1$. The design matrix therefore simplifies to
\begin{equation}\label{eq:designmtx}
P = \left(\begin{array}{cc} p_{00} & 1-p_{00} \\ 1-p_{11} & p_{11}\end{array}\right),
\end{equation}
where $p_{00},p_{11} \in [0,1]$.
\end{definition}

As $\pi$ is the true proportion of individuals in the population possessing the sensitive attribute, we can calculate the probability mass function of each $X_i$:
\begin{subequations}\label{eq:pmf}
\begin{gather}
\begin{aligned}
\Prob(X_i = 0) &= (1-\pi) p_{00} + \pi (1-p_{11})\\
&= p_{00} - \pi(p_{00} + p_{11} - 1),
\end{aligned}\\[1em]
\begin{aligned}
\Prob(X_i = 1) &= \pi p_{11} + (1-\pi)(1-p_{00})\\
&= 1-p_{00} + \pi(p_{00} + p_{11} - 1).
\end{aligned}
\end{gather}
\end{subequations}

\textbf{Remark:} Direct questioning corresponds to the case where $p_{00}=p_{11}=1$.

\subsection{Estimator, Bias and Error}
Having presented the RR mechanism previously, we now need to establish an estimator of $\pi$ from the parameters of the mechanism, $p_{00}$ and $p_{11}$, and from the distribution of randomised responses, namely $\frac{N}{n}$. We first establish a maximum likelihood estimator (MLE) for the mechanism and then examine its bias and error.

\begin{theorem}\label{th:mle}
Let $p_{00} + p_{11} \neq 1$. Then the MLE for $\pi$ of the randomised response mechanism given by (\ref{eq:designmtx}) is
\begin{equation}\label{eq:generalmle}
\hat{\Pi}(p_{00},p_{11}) = \frac{p_{00}-1}{p_{00} + p_{11} - 1} + \frac{N}{(p_{00} + p_{11} - 1)n}.
\end{equation}
\end{theorem}

	\begin{proof}
	Let us first index the sample so that $X_i=1$ for each $i\le N$, and $X_i=0$ for each $i>N$. Then the likelihood $L$ of the sample is
	$$L = \Prob(X_i=1)^N \Prob(X_i=0)^{n-N}.$$
	The log-likelihood is
	$$\log(L) = N \log \Prob(X_i=1) + (n-N) \log \Prob(X_i=0),$$
	whose derivatives are
	\begin{align*}
	\frac{\partial \log(L)}{\partial \pi} &= \frac{N}{\Prob(X_i=1)}\frac{\partial \Prob(X_i=1)}{\partial \pi}+\frac{n-N}{\Prob(X_i=0)}\frac{\partial \Prob(X_i=0)}{\partial \pi},\\
	\frac{\partial^2 \log(L)}{\partial \pi^2} &= -\frac{N}{\Prob(X_i=1)^2}\left(\frac{\partial \Prob(X_i=1)}{\partial \pi}\right)^2-\frac{n-N}{\Prob(X_i=0)^2}\left(\frac{\partial \Prob(X_i=0)}{\partial \pi}\right)^2.
	\end{align*}
	We note that $\frac{\partial^2 \log(L)}{\partial \pi^2}<0$, hence the maximum of $\log(L)$ occurs when $\frac{\partial \log(L)}{\partial \pi}=0$. Solving for $\pi$ completes the proof.
	\end{proof}

We note the following standard identity in probability and statistics,
\begin{equation}\label{eq:var}
\Var(Y) = \mathbb{E}[Y^2] - \mathbb{E}[Y]^2,
\end{equation}
for any random variable $Y$. We now calculate the bias and error of $\hat{\Pi}$.  We use the variance of the estimator to characterise error in line with conventional practice.  Similarly by convention, we characterise the bias of an estimator as its expected deviation from the quantity it is estimating (\ie\ $\mathbb{E}[\hat{\Pi}-\pi]$). We remind the reader of the dependence of $\Var(\hat{\pi})$ on $\pi$ by writing $\Var(\hat{\Pi}| \pi)$.

\begin{corollary}
The MLE $\hat{\Pi}$ constructed in Theorem~\ref{th:mle} is unbiased and has error
\begin{equation}\label{eq:genvar}
\Var(\hat{\Pi}(p_{00},p_{11})| \pi) = \frac{\frac{1}{4} - \left(p_{00}-\frac{1}{2}-\pi(p_{00} + p_{11} - 1)\right)^2}{(p_{00}+p_{11}-1)^2 n}.
\end{equation}
\end{corollary}

	\begin{proof}
	Since the survey we are conducting is by uniform random sampling with replacement, $N$ is a sum of independent and identically distributed random variables.  Therefore, $\mathbb{E}[N] = n \mathbb{E}[X_i]$ and $\Var(N) = n \Var(X_i)$.
	
	Since $X_i \in \{0,1\}$, it can be shown that $\mathbb{E}[X_i] = \mathbb{E}[X_i^2] = \Prob(X_i=1) = 1-p_{00} + \pi(p_{00} + p_{11}-1)$. Hence,
	\begin{align*}
	\mathbb{E}[\hat{\Pi}] &= \frac{p_{00}-1}{p_{00} + p_{11} - 1} + \frac{\mathbb{E}[N]}{(p_{00} + p_{11} - 1)n}\\
	&= \frac{p_{00}-1}{p_{00} + p_{11} - 1} + \frac{\mathbb{E}[X_i]}{p_{00} + p_{11} - 1}\\
	&= \pi,
	\end{align*}
	and so $\hat{\Pi}$ is unbiased as claimed.
	
	Secondly,
	\begin{align*}
	\Var(\hat{\Pi}| \pi) &= \frac{\Var(N)}{(p_{00}+p_{11}-1)^2 n^2}\\
	&= \frac{\Var(X_i)}{(p_{00}+p_{11}-1)^2 n}\\
	&= \frac{\mathbb{E}[X_i^2]-\mathbb{E}[X_i]^2}{(p_{00}+p_{11}-1)^2 n}\\
	&= \frac{\Prob(X_i=1)\Prob(X_i=0)}{(p_{00}+p_{11}-1)^2 n},
	\end{align*}
	which can be simplified to (\ref{eq:genvar}).
	\end{proof}

\begin{subequations}
When conducting a survey on a population, it is often useful and necessary to estimate the margin of error of the estimate on a sample.  For a confidence level $c \in [0,1]$, the \emph{margin of error} of a sample is given by $\omega \ge 0$, where
\begin{equation}
\Prob(|\hat{\Pi}-\pi| \le \omega) \ge c.
\end{equation}
In practical applications, a 95\% confidence interval is typically used \cite{Jac05}.  In the absence of any additional information on the distribution of $\hat{\Pi}$, Chebyshev's inequality can be used to derive a general, but conservative, margin of error, assuming $\hat{\Pi}$ has finite variance. In such a scenario, the margin of error of a sample is given to be $4.5\sigma$, where the standard deviation $\sigma$ is given by $\sqrt{\Var(\hat{\Pi}|\pi)}$, since
\begin{equation}\label{eq:moe1}
\Prob\left(|\hat{\Pi}-\pi| \le 4.5\sqrt{\Var(\hat{\Pi}|\pi)}\right) \ge 0.95.
\end{equation}

In many practical situations, the central limit theorem is invoked to determine heuristically a margin of error.  For a random variable $G$ that is normally distributed with mean $\mu$ and variance $\sigma^2$, we have
\begin{equation}\label{eq:moe2}
\Prob(|G-\mu| \le 1.96 \sigma) \ge 0.95,
\end{equation}
hence $1.96 \sigma$ is typically taken as the margin of error in such scenarios \cite{Jac05}.  However, this non-rigorous approach only gives a loose representation of the margin of error, given that the guarantee of the central limit theorem only applies in the limit as the sample size $n$ approaches infinity.
\end{subequations}

Due to this variability in defining the margin of error of a sample, we only focus on determining the error of the estimator, $\Var(\hat{\Pi}|\pi)$, in this paper.  This error can be used to calculate the margin of error for a particular application, as outlined above.

\subsection{Warner's RR model}\label{sc:warner}
Warner's model \cite{War65} is a specific case of the generalised model introduced in Section~\ref{sc:genrr}. Warner proposed that surveyors would present respondents with a spinner which they would spin in private to decide which one of two questions to answer. The spinner would point to a question (\eg\ ``Have you ever cheated on your spouse\slash partner?'') with probability $p_w$, and to the complement of that question (\eg\ ``Have you always been faithful to your spouse\slash partner?'') with probability $1-p_w$.  Respondents would then be asked to answer the chosen question truthfully, but without revealing which question they were answering.  As before, $x_i$ denotes the truthful response of respondent $i$, while $X_i$ denotes the randomised response, as determined by the process outlined above.

Warner's model corresponds to the case where $p_{00} = p_{11} = p_w$. We denote by $P_w$ the design matrix of Warner's model, which is given by
$$P_w = \left(\begin{array}{cc} p_w & 1-p_w \\ 1-p_w & p_w\end{array}\right),$$
while the probability mass function of each $X_i$ is defined as
\begin{align*}
\Prob(X_i = 0) &= p_w - \pi(2p_w - 1),\\
\Prob(X_i = 1) &= 1-p_w + \pi(2p_w - 1).
\end{align*}

Using the same unbiased MLE in (\ref{eq:generalmle}), we denote by $\hat{\Pi}_w$ the estimator for Warner's model and, by (\ref{eq:genvar}), find its error to be
\begin{equation}\label{eq:warnererror}
\Var(\hat{\Pi}_w(p_w)| \pi) = \frac{\frac{1}{4} - \left(p_w-\frac{1}{2}-\pi(2p_w - 1)\right)^2}{(2p_w-1)^2 n}.
\end{equation}

\section{Differential Privacy}\label{sc:dp}

Differential privacy was first proposed by Dwork in 2006 \cite{Dwo06} as a way to measure the level of privacy achieved when publishing data.  Using the same notation as in \cite{HLM15}, we denote by $D^m$ the space of all $m$-row datasets (let $D$ be the space of each row) and by $\md \in D^m$ a dataset in this space. We then denote by $X_\md: \Omega \to D^n$ a randomised version of $\md$.

If $D$ is assumed to be discrete, the mechanism $X_\md$ is said to satisfy ($\epsilon$,$\delta$)-differential privacy if
\begin{equation}\label{eq:dp}
\Prob(X_\md \in A) \le e^\epsilon \Prob(X_{\md^\prime} \in A) + \delta,
\end{equation}
for each $\md, \md^\prime \in D^m$ that differ in exactly one row (\ie\ there exists exactly one $j \in [m]$ such that $d_j \neq d_j^\prime$) and for each subset $A \subset D^m$.

This set-up simplifies in the case of randomised response introduced in Section~\ref{sc:prelim}.  Firstly, the datasets contain only one row ($m=1$), and the row-space is $\{0,1\}$.  We are therefore only required to show that (\ref{eq:dp}) holds for $\md \neq \md^\prime \in \{0,1\}$ and for $A=\{0\}, \{1\}$. Formally, ($\epsilon$,$\delta$)-differential privacy is satisfied if
\begin{equation}
\Prob(X_i = j) \le e^\epsilon \Prob(X_k = j) + \delta,
\end{equation}
for any $i, k \in [n]$ and $j \in \{0,1\}$.

For the RR mechanism given by (\ref{eq:designmtx}) to satisfy {\ed}-differential privacy, we require the following to hold:
\begin{subequations}\label{eq:dp0}
\begin{align}
p_{11} &\le e^\epsilon (1-p_{00}) + \delta,\label{eq:dp1}\\
p_{00} &\le e^\epsilon (1-p_{11}) + \delta,\label{eq:dp2}\\
1-p_{00} &\le e^\epsilon p_{11} + \delta,\nonumber\\
1-p_{11} &\le e^\epsilon p_{00} + \delta.\nonumber
\end{align}
\end{subequations}

We can now define the set of pairs $(p_{00}, p_{11})$ that correspond to a RR mechanism which satisfies {\ed}-differential privacy.

\begin{definition}[Region of Feasibility]
A RR mechanism, given by (\ref{eq:designmtx}), satisfies {\ed}-differential privacy if $(p_{00},p_{11})\in\mathcal{R}$, where $\mathcal{R} \subset \R^2$ is defined as
\begin{equation}
\mathcal{R}=\left\{(p_{00},p_{11})\in\R^2:\begin{array}{l} p_{00},p_{11} \in [0,1],\\
p_{00}\le e^\epsilon(1-p_{11})+\delta,\\
p_{11}\le e^\epsilon(1-p_{00})+\delta,\\
1-p_{11} \le e^\epsilon p_{00}+\delta,\\
1-p_{00} \le e^\epsilon p_{11}+\delta.\end{array}\right\}.
\end{equation}
\end{definition}

We consider the case where $p_{00}+p_{11} > 1$.  Note that the estimator error, and hence the optimal mechanism, is undefined when $p_{00}+p_{11} = 1$. If $p_{00} + p_{11} < 1$, we permute all responses such that $X^\prime_i = 1-X_i$.  This corresponds to the columns of the design matrix being swapped, giving $p^\prime_{00}=1-p_{00}$ and $p^\prime_{11}=1-p_{11}$, hence $p^\prime_{00} + p^\prime_{11}=2-p_{00}-p_{11}>1$.  We can therefore assume $p_{00} + p_{11} > 1$ without loss of generality.

When $p_{00} + p_{11} > 1$, we note that (i) $1-p_{11} < p_{00} \le e^\epsilon (1-p_{11})+\delta < e^\epsilon p_{00}+\delta$ and (ii) $1-p_{00} < p_{11} \le e^\epsilon (1-p_{00})+\delta < e^\epsilon p_{11}+\delta$. Hence, the region of feasibility simplifies to $\mathcal{R}^\prime$ as follows:
\begin{align*}
\mathcal{R}^\prime &= \left\{ (p_{00},p_{11})\in \mathcal{R} : p_{00} + p_{11} > 1 \right\}\\
&= \left\{ (p_{00},p_{11})\in\R:\begin{array}{l} p_{00}, p_{11} \le 1, \\
p_{00} + p_{11} > 1 ,\\
p_{00}\le e^\epsilon(1-p_{11})+\delta,\\
p_{11}\le e^\epsilon(1-p_{00})+\delta.\end{array}\right\}.
\end{align*}

Furthermore, we denote by $\mathcal{R}^\double$ the boundary of $\mathcal{R}^\prime$ which satisfies at least one of inequalities (\ref{eq:dp0}):
$$\mathcal{R}^\double = \mathcal{R}^\prime \setminus \left\{(p_{00}, p_{11})\in \R : \begin{array}{l} p_{00} < e^\epsilon(1-p_{11})+\delta,\\
p_{11} < e^\epsilon(1-p_{00})+\delta.\end{array}\right\}.$$
The set $\mathcal{R}^\double$ therefore consists of the union of two line segments in the unit square, where (\ref{eq:dp1}) and (\ref{eq:dp2}) are tight.

We are therefore looking to find the RR mechanism which minimises estimator error, while still being {\ed}-differentially private.  Hence, we seek to find
\begin{equation}
\argmin_{(p_{00},p_{11})\in\mathcal{R}^\prime} \Var\left(\left.\hat{\Pi}(p_{00},p_{11})\right| \pi\right).
\end{equation}

\section{Preliminary Results}\label{sc:prelimresults}

We begin by presenting two results which will be of use later in the paper.  The first result concerns the non-negativity of a non-linear function on the unit square.

\begin{lemma}\label{lm:fxy}
Let $f:\R \times \R \to \R$ be defined by
$$f(x, y)=2xy-x-y+1.$$
Then, $f(x,y)\ge 0$ for all $x, y \in [0,1]$.

Furthermore,
$$\argmin_{x, y \in [0,1]} f(x,y) = \{(0, 1), (1, 0)\}.$$
\end{lemma}

	\begin{proof}
	Let's first consider $\min_{x\in [0,1]} f(x,y)$:
	\begin{align}
	\min_{x \in [0,1]} f(x,y) &= \min_{x \in [0,1]} (2xy-x)-y+1\nonumber\\
	&= \min_{x \in [0,1]} \left((2y-1)x\right)-y+1\nonumber\\
	&=\begin{cases} y & \text{if } y \le \frac{1}{2},\\
	1-y & \text{if } y > \frac{1}{2}.\label{eq:fxy}
	\end{cases}
	\end{align}
	It follows that
	$$\min_{y \in [0,1]} \left(\min_{x \in [0,1]} f(x,y)\right) = 0.$$
	By symmetry of $f$, it also follows that
	$$\min_{x \in [0,1]} \left(\min_{y \in [0,1]} f(x,y)\right) = 0,$$
	hence $f(x,y) \ge 0$ for all $x, y \in [0,1]$.
	
	We note that $f(1,0)=f(0,1)=0$, and by (\ref{eq:fxy}) we see that these values uniquely minimise $f(x,y)$ for all $x,y \in [0,1]$.
	\end{proof}

In the second result of this section we prove that an optimal mechanism exists on $\mathcal{R}^\double$ (\ie\ on the boundary of $\mathcal{R}^\prime$ where at least one of inequalities (\ref{eq:dp0}) is tight), and additionally that when $\pi \in (0,1)$, optimal mechanisms only occur on $\mathcal{R}^\double$.

\begin{lemma}\label{lm:eddp}
Let $p_{00} + p_{11} > 1$. Then there exists $(p_{00}^*, p_{11}^*) \in \argmin_{\mathcal{R}^\prime} \Var(\hat{\Pi}|\pi)$ such that $(p_{00}^*, p_{11}^*) \in \mathcal{R}^\double$.

Furthermore, when $0 < \pi < 1$, $\argmin_{\mathcal{R}^\prime} \Var(\hat{\Pi}|\pi) \subseteq \mathcal{R}^\double$.
\end{lemma}

	\begin{proof}
	Let's consider $\frac{\partial \Var(\hat{\Pi}| \pi)}{\partial p_{00}}$ and $\frac{\partial \Var(\hat{\Pi}| \pi)}{\partial p_{11}}$.
	
	Firstly, after some rearranging\slash manipulation,
	$$\frac{\partial \Var(\hat{\Pi}| \pi)}{\partial p_{11}} = -\frac{2p_{00}(1-p_{00})(1-\pi)+\pi(2p_{00}p_{11}-p_{00}-p_{11}+1)}{(p_{00}+p_{11}-1)^3n}.$$
	By Lemma~\ref{lm:fxy}, we know that $2p_{00}p_{11}-p_{00}-p_{11}+1 \ge 0$, and since $p_{00}+p_{11} - 1 > 0$ by hypothesis, we conclude that $\frac{\partial \Var(\hat{\Pi}| \pi)}{\partial p_{11}} \le 0$.
	
	We further note that $2p_{00}p_{11}-p_{00}-p_{11}+1 > 0$ by Lemma~\ref{lm:fxy}, since the assumption that $p_{00}+p_{11} > 1$ means $p_{00},p_{11} > 0$.  Hence $\frac{\partial \Var(\hat{\Pi}| \pi)}{\partial p_{11}} = 0$ only when $\pi = 0$ and $p_{00} = 1$.  Equivalently,
	\begin{equation}\label{eq:min1}
	\frac{\partial \Var(\hat{\Pi}| \pi)}{\partial p_{11}} < 0 \text{ when } \pi > 0 \text{ or } p_{00} < 1.
	\end{equation}
	
	Secondly, after some rearranging\slash manipulation,
	$$\frac{\partial \Var(\hat{\Pi}| \pi)}{\partial p_{00}} = 
	-\frac{(2p_{00}p_{11}-p_{00}-p_{11}+1)(1-\pi)+2p_{11}\pi(1-p_{11})}{(p_{00}+p_{11}-1)^3 n}.$$
	Since, by assumption, we have $2p_{00}p_{11}-p_{00}-p_{11}+1 \ge 0$ and since $p_{11} \in [0,1]$, we see that $\frac{\partial \Var(\hat{\Pi}| \pi)}{\partial p_{00}} \le 0$.
	
	Similar to the reasoning above, since $2p_{00}p_{11}-p_{00}-p_{11}+1 > 0$ and $p_{11} > 0$, $\frac{\partial \Var(\hat{\Pi}| \pi)}{\partial p_{00}} = 0$ only when $\pi = 1$ and $p_{11} = 1$. Equivalently,
	\begin{equation}\label{eq:min2}
	\frac{\partial \Var(\hat{\Pi}| \pi)}{\partial p_{00}} < 0 \text{ when } \pi < 1 \text{ or } p_{11} < 1.
	\end{equation}
	
	Since $\frac{\partial \Var(\hat{\Pi}| \pi)}{\partial p_{00}} \le 0$ and $\frac{\partial \Var(\hat{\Pi}| \pi)}{\partial p_{11}} \le 0$, there exists a mechanism on the boundary of $\mathcal{R}^\prime$ which minimises the estimator error, \ie
	\begin{equation}\label{eq:dpconstr}
	\partial\mathcal{R}^\prime \cap \left(\argmin_{(p_{00},p_{11})\in\mathcal{R}^\prime} \Var(\hat{\Pi}(p_{00},p_{11})| \pi)\right) \neq \emptyset.
	\end{equation}
	
	However, if $0 < \pi < 1$, we see from (\ref{eq:min1}) and (\ref{eq:min2}) that $\frac{\partial \Var(\hat{\Pi}| \pi)}{\partial p_{00}} < 0$ and $\frac{\partial \Var(\hat{\Pi}| \pi)}{\partial p_{11}} < 0$. Hence,
	\begin{equation}\label{eq:dpconstr2}
	\argmin_{(p_{00},p_{11})\in\mathcal{R}^\prime} \Var(\hat{\Pi}(p_{00},p_{11})| \pi) \subseteq \partial \mathcal{R}^\prime,
	\end{equation}
	\ie\ the optimal mechanisms \emph{only} occur on the boundary of $\mathcal{R}^\prime$.

	Finally, suppose $(p_{00},p_{11}) \in \partial\mathcal{R}^\prime$, but neither of the inequalities in (\ref{eq:dp0}) are tight.  Then there exist $\Delta_0,\Delta_1 \ge 0$, $\Delta_0+\Delta_1>0$ where $(p_{00}+\Delta_0,p_{11}+\Delta_1)\in\partial\mathcal{R}^\prime$, but because $\frac{\partial \Var(\hat{\Pi}| \pi)}{\partial p_{00}} \le 0$ and $\frac{\partial \Var(\hat{\Pi}| \pi)}{\partial p_{11}} \le 0$, then $\Var(\hat{\Pi}(p_{00},p_{11})| \pi) \ge \Var(\hat{\Pi}(p_{00}+\Delta_0,p_{11}+\Delta_1)| \pi)$. Hence minimal error is achieved when at least one of the inequalities (\ref{eq:dp0}) is tight, \ie
	\begin{align*}
	&\argmin_{(p_{00},p_{11})\in\mathcal{R}^\prime} \Var(\hat{\Pi}(p_{00},p_{11})| \pi) \subseteq \mathcal{R}^\double.\qedhere
	\end{align*}
	\end{proof}

For the remainder of this paper, we assume $\pi \in (0,1)$.  Note that the results on optimal mechanisms still hold for $\pi \in [0,1]$, however these optima may not be unique.

\section{Optimal Mechanism for \texorpdfstring{$\epsilon$-}{Strict }Differential Privacy}\label{sc:optsdp}
We have already established that the parameters for the optimal {\ed}-differentially private mechanism lie on $\mathcal{R}^\double$.  We now examine the case of $\epsilon$-differential privacy, where $\delta=0$, with the additional assumption that $\epsilon > 0$.

\begin{theorem}\label{th:edp}
Let $\pi \in (0,1)$, $p_{00} + p_{11} > 1$  and $\epsilon > 0$.  The $\epsilon$-differentially private RR mechanism which minimises estimator error is given by the design matrix
$$P_\epsilon = \left(\begin{array}{cc} \frac{e^\epsilon}{e^\epsilon+1} & \frac{1}{e^\epsilon+1} \\ \frac{1}{e^\epsilon+1} & \frac{e^\epsilon}{e^\epsilon+1} \end{array}\right).$$
\end{theorem}

	\begin{proof}
	By Lemma~\ref{lm:eddp}, we know that the parameters $(p_{00}, p_{11})$ of the optimal mechanism exist on the boundary of $\mathcal{R}^\prime$, with at least one of the inequalities (\ref{eq:dp0}) tight.  We now separately consider the cases where (\ref{eq:dp1}) and (\ref{eq:dp2}) are tight.  By hypothesis, $\delta=0$ and $\epsilon \neq 0$. 
	
	\begin{enumerate}
	\item (\ref{eq:dp1}) tight: $p_{11}=e^\epsilon(1-p_{00})$, constrained by $p_{11}\ge 0$ and $p_{00}\le e^\epsilon(1-p_{11})$. By (\ref{eq:dp2}) and since $p_{00} = 1-e^{-\epsilon}p_{11}$, we have
	\begin{align*}
	e^\epsilon p_{11} &\le e^\epsilon - p_{00}\\
	&=e^\epsilon-(1-e^{-\epsilon}p_{11})\\
	&=e^\epsilon-1+e^{-\epsilon}p_{11},
	\end{align*}
	which we rewrite as
	$$p_{11}(e^\epsilon - e^{-\epsilon})\le e^\epsilon -1,$$
	and noting that $e^{2\epsilon}-1=(e^\epsilon-1)(e^\epsilon+1)$, we see that
	\begin{align*}
	p_{11} &\le \frac{e^\epsilon-1}{e^{-\epsilon}(e^{2\epsilon}-1)}\\
	&= \frac{e^\epsilon}{e^\epsilon+1}.
	\end{align*}
	
	We are therefore considering $\Var(\hat{\Pi}(p_{00},p_{11})| \pi)$ on the line $p_{00} = 1-e^{-\epsilon}p_{11}$ for $0 \le p_{11} \le \frac{e^\epsilon}{e^\epsilon+1}$.  We parametrise this line as follows, where $0 < t \le 1$, $p_{00} = r(t)$ and $p_{11} = s(t)$ (we require $t>0$ since $p_{00}+p_{11} > 1$):
	\begin{equation}
	\begin{aligned}\label{eq:edpparam}
	r(t) &= (1-t)+\frac{e^\epsilon}{1+e^\epsilon}t=1-e^{-\epsilon}s(t),\\
	s(t) &= \frac{e^\epsilon}{1+e^\epsilon}t.
	\end{aligned}
	\end{equation}
	For simplicity, we let $\hat{\Pi}(r(t), s(t))=\hat{\Pi}_1(t)$. After some manipulation, we see that
	$$\frac{\partial \Var(\hat{\Pi}_1(t)| \pi)}{\partial t} = -\frac{(1+e^\epsilon)(1+\pi(e^\epsilon-1))}{(e^\epsilon-1)^2t^2n},$$
	and noting that $e^\epsilon > 1$, we see that $\frac{\partial \Var(\hat{\Pi}_1(t)| \pi)}{\partial t}<0$.  Hence,
	\begin{equation}\label{eq:argmin1}
	\argmin_{t \in (0,1]} \Var(\hat{\Pi}_1(t)| \pi) = \{1\}.
	\end{equation}
	
	\item (\ref{eq:dp2}) tight: By symmetry of the equations (\ref{eq:dp0}), we simply let $p_{00}=s(t)$ and $p_{11}=r(t)$.  By examining (\ref{eq:pmf}) and (\ref{eq:genvar}), we see that
	$$\Var(\hat{\Pi}(p_{00}, p_{11})| 1-\pi) = \Var(\hat{\Pi}(p_{11}, p_{00})| \pi),$$
	and by letting $\hat{\Pi}(s(t), r(t))=\hat{\Pi}_2(t)$, we get
	$$\frac{\partial \Var(\hat{\Pi}_2(t)| \pi)}{\partial t} = -\frac{(1+e^\epsilon)(1+(1-\pi)(e^\epsilon-1)}{(e^\epsilon-1)^2t^2n}.$$
	Again it follows that $\frac{\partial \Var(\hat{\Pi}_2(t)| \pi)}{\partial t}<0$, and so
	\begin{equation}\label{eq:argmin2}
	\argmin_{t \in (0,1]} \Var(\hat{\Pi}_2(t)| \pi) = \{1\}.
	\end{equation}
	\end{enumerate}
	
	By (\ref{eq:dpconstr2}), (\ref{eq:argmin1}) and (\ref{eq:argmin2}), we can now conclude that
	\begin{equation*}
	\argmin_{(p_{00},p_{11})\in\mathcal{R}^\prime} \Var(\hat{\Pi}(p_{00},p_{11})| \pi)= \left\{\left(\frac{e^\epsilon}{e^\epsilon+1}, \frac{e^\epsilon}{e^\epsilon+1}\right)\right\},
	\end{equation*}
	and so the result follows.
	\end{proof}

\textbf{Remark:} When $\epsilon=0$, all rows of the design matrix must be identical, \ie\ $p_{00}=1-p_{11}$ and $p_{11}=1-p_{00}$.  This gives $p_{00}+p_{11} = 1$, leading to an unbounded estimator error (\ref{eq:genvar}).  In practical terms, \num{0}-differential privacy enforces the same output distribution for every respondent, hence nothing meaningful can be learned.

\section{Optimal Mechanism for \texorpdfstring{{\ed}-}{Relaxed }Differential Privacy}\label{sc:optrdp}

Let's now consider the optimal mechanism for {\ed}-differential privacy.  We parametrise $\mathcal{R}^\double$ as follows.  If we let
\begin{equation}
\begin{aligned}\label{eq:eddpparam}
r_\delta(t) &= \left(1+e^{-\epsilon}\delta\right)(1-t)+\frac{e^\epsilon+\delta}{e^\epsilon+1}t,\\
&=1-e^{-\epsilon}(s_\delta(t)-\delta),\\
s_\delta(t) &= \frac{e^\epsilon+\delta}{e^\epsilon+1} t,
\end{aligned}
\end{equation}
for $t\in [0,1]$, then the boundary where (\ref{eq:dp1}) holds is parametrised by $p_{00} = r_\delta(t)$ and $p_{11} = s_\delta(t)$; by symmetry, the boundary where (\ref{eq:dp2}) holds is parametrised by $p_{00} = s_\delta(t)$ and $p_{11} = r_\delta(t)$.

We note that $t=1$ denotes an extreme point of $\mathcal{R}^\prime$ (and $\mathcal{R}^\double$), the point at which both inequalities (\ref{eq:dp0}) are tight.  Here $p_{00} = p_{11} = r_\delta(1) = s_\delta(1)= \frac{e^\epsilon+\delta}{e^\epsilon+1}$.

\subsection{Preliminary Lemmas}

Before proceeding to the main result of this section, we first present a collection of lemmas for later use.  The first result states that the minimal variance of $\hat{\Pi}$ on $\mathcal{R}^\double$ will occur at one of its extreme points (\ie\ at one of the endpoints of the two line segments which comprise $\mathcal{R}^\double$).

\begin{lemma}\label{lm:argminendpts}
Let $r_\delta$ and $s_\delta$ be given by (\ref{eq:eddpparam}), let $\delta > 0$ and let $a \le b \in [0,1]$. Then,
$$\argmin_{t \in [a,b]} \Var(\hat{\Pi}(r_\delta(t), s_\delta(t))| \pi) \subseteq \{a,b\}.$$
\end{lemma}

	\begin{proof}
	For simplicity, we denote $\hat{\Pi}(r_\delta(t), s_\delta(t))$ by $\hat{\Pi}_{1,\delta}(t)$.
	
	By some manipulation, it can be shown that the numerator of $\frac{\partial \Var(\hat{\Pi}_{1,\delta}(t)| \pi)}{\partial t}$ is linear in $t$, hence it has at most one root at
$$t = \frac{\delta(1+e^\epsilon)(2e^\epsilon+2\delta-1-\pi(e^\epsilon+2\delta-1))}{(e^\epsilon+\delta)(e^\epsilon+2\delta-1)(1+(e^\epsilon-1)\pi)}.$$

	By substitution, we find that
	$$\frac{\partial^2 \Var(\hat{\Pi}_{1,\delta}(t)| \pi)}{\partial t^2} = -\frac{(e^\epsilon+\delta)^2(e^\epsilon+2\delta-1)^4(1+(e^\epsilon-1)\pi)^4}{8e^{2\epsilon}\delta^3(e^\epsilon+\delta-1)^3(1+e^\epsilon)^2n},$$
	when $\frac{\partial \Var(\hat{\Pi}_{1,\delta}(t)| \pi)}{\partial t}=0$. By inspection, and since $\delta > 0$, we see that $\frac{\partial^2 \Var(\hat{\Pi}_{1,\delta}(t)| \pi)}{\partial t^2} < 0$ when $\frac{\partial \Var(\hat{\Pi}_{1,\delta}(t)| \pi)}{\partial t}=0$, and so this point is the maximum of $\Var(\hat{\Pi}_{1,\delta}(t)| \pi)$.  Hence, the minimum of $\Var(\hat{\Pi}_{1,\delta}(t)| \pi)$ cannot occur at a mid-point of an interval.  The result follows.
	\end{proof}

We next show that the error of $\hat{\Pi}$ along the boundary constrained by (\ref{eq:dp1}) is uniformly greater than along the boundary constrained by (\ref{eq:dp2}) when $\pi \le \frac{1}{2}$.

\begin{lemma}\label{lm:argminrssr}
Let $r_\delta$ and $s_\delta$ be given by (\ref{eq:eddpparam}) and let $\delta > 0$. Then, when $\pi \le \frac{1}{2}$,
$$\Var(\hat{\Pi}(r_\delta(t), s_\delta(t))| \pi) \le \Var(\hat{\Pi}(s_\delta(t), r_\delta(t))| \pi),$$
for $t \in [0,1]$.

Conversely, if $\pi \ge \frac{1}{2}$, then
$$\Var(\hat{\Pi}(r_\delta(t), s_\delta(t))| \pi) \ge \Var(\hat{\Pi}(s_\delta(t), r_\delta(t))| \pi),$$
for $t \in [0,1]$.
\end{lemma}

	\begin{proof}
	After manipulation of the terms, we can show that
	$$\Var(\hat{\Pi}(r_\delta(t), s_\delta(t))| \pi) - \Var(\hat{\Pi}(s_\delta(t), r_\delta(t))| \pi) = 
	-\frac{(e^\epsilon+1)(e^\epsilon+\delta)(1-2\pi)(1-t)}{(e^\epsilon(e^\epsilon-1)t+\delta(1-t+e^\epsilon(1+t)))n}.$$
	We see that $1-2\pi \ge 0$ when $\pi \le \frac{1}{2}$, and $1-2\pi \le 0$ when $\pi \ge \frac{1}{2}$, and, since $t\in[0,1]$ and $\delta > 0$, the result follows.
	\end{proof}

Finally, we present $t_0(\epsilon, \delta)$ as the $t$-value which gives the endpoints of the line segments of $\mathcal{R}^\double$ at the boundary of the unit square.
	
\begin{lemma}
Define $t_0:\R \times \R \to [0,1]$ by
$$t_0(\epsilon,\delta) = \frac{\delta(e^\epsilon+1)}{e^\epsilon+\delta},$$
then,
$$(r_\delta(t_0(\epsilon,\delta)), s_\delta(t_0(\epsilon, \delta))) \in \partial \mathcal{R}^\prime.$$
\end{lemma}

	\begin{proof}
	By explicit calculation,
	\begin{align*}
	r_\delta(t_0(\epsilon,\delta)) &= 1,\\
	s_\delta(t_0(\epsilon,\delta)) &= \delta.
	\end{align*}
	
	By definition, it follows that $(1, \delta) \in \mathcal{R}^\prime \cup \partial \mathcal{R}^\prime$, and since $p_{00} \le 1$ is a boundary of $\{(p_{00},p_{11}) \in \mathcal{R}^\prime\}$, it follows that $(1, \delta) \in \partial \mathcal{R}^\prime$.
	\end{proof}

\textbf{Remark:} When $\delta = 0$, $\left(r_\delta(t_0(\epsilon,\delta)), s_\delta(t_0(\epsilon, \delta))\right) \notin \mathcal{R}^\prime$, since we require $r_\delta + s_\delta > 1$.

\textbf{Remark:} By linearity, it follows that $(r_\delta(t), s_\delta(t)) \in \mathcal{R}^\prime$ for all $t_0(\epsilon,\delta) < t \le 1$, and that $(r_\delta(t), s_\delta(t)) \notin \mathcal{R}^\prime$ when $t < t_0(\epsilon,\delta)$.  

\subsection{Main Result}

We now present the main results of this paper, which establish the optimal {\ed}-differentially private RR mechanism(s).  The following results assume $\delta > 0$; the optimal mechanism when $\delta = 0$ was presented in Theorem~\ref{th:edp}.  Note that we continue to assume $\pi \in (0,1)$ to ensure uniqueness of the optima.

The following theorem establishes the optimal RR mechanism(s) when $\pi \le \frac{1}{2}$.

\begin{theorem}\label{th:opteddp}
Let $\delta > 0$ and $0 < \pi \le \frac{1}{2}$, and define $g:\R\times\R\to\R$ by
\begin{equation}\label{eq:g_ed}
g(\epsilon,\delta) = \frac{\delta(e^\epsilon+\delta)}{(e^\epsilon+2\delta-1)^2}.
\end{equation}
Then, for $r_\delta$ and $s_\delta$ given by (\ref{eq:eddpparam}),
$$\argmin_{(p_{00},p_{11}) \in \mathcal{R}^\prime} \Var(\hat{\Pi}(p_{00},p_{11})| \pi)
=\begin{cases} \{(r_\delta(t_0),s_\delta(t_0))\}, & \text{if } g(\epsilon,\delta) > \pi,\\
\{(r_\delta(1),s_\delta(1))\}, & \text{if } g(\epsilon,\delta) < \pi,\\
\{(r_\delta(t_0),s_\delta(t_0)), (r_\delta(1),s_\delta(1))\}, & \text{if } g(\epsilon,\delta) = \pi.\end{cases}$$
where $t_0 = t_0(\epsilon,\delta)$.
\end{theorem}

	\begin{proof}
	By Lemmas~\ref{lm:eddp}, \ref{lm:argminendpts} and \ref{lm:argminrssr}, we know that when $0 < \pi \le \frac{1}{2}$ and $\delta > 0$,
	$$\argmin_{(p_{00}, p_{11})\in \mathcal{R}^\prime} \Var(\hat{\Pi}(p_{00}, p_{11}) | \pi) \subseteq \{(r_\delta(t_0), s_\delta(t_0)), (r_\delta(1), s_\delta(1))\}.$$
	
	We are therefore considering two candidate points, which can be shown to resolve to
	\begin{align*}
	r_\delta(t_0) &= 1, & s_\delta(t_0) &= \delta,\\
	r_\delta(1) &= \frac{e^\epsilon + \delta}{e^\epsilon+1}, & s_\delta(1) &= \frac{e^\epsilon+\delta}{e^\epsilon+1}.
	\end{align*}
	We are therefore seeking to determine the sign of
	\begin{equation}\label{eq:optvarcomp}
	\Var(\hat{\Pi}(1, \delta) | \pi )-\Var\left(\left.\hat{\Pi}\left(\frac{e^\epsilon + \delta}{e^\epsilon+1},\frac{e^\epsilon + \delta}{e^\epsilon+1}\right) \right| \pi\right).
	\end{equation}
	
	After some manipulation, we can show that (\ref{eq:optvarcomp}) simplifies to
	$$\frac{(1-\delta)(\pi(e^\epsilon+2\delta-1)-\delta(e^\epsilon+\delta))}{\delta(e^\epsilon+2\delta-1)^2n},$$
	and we note that its denominator is strictly positive since $\delta > 0$. Note additionally that (\ref{eq:optvarcomp}) simplifies to zero when $\delta = 1$, which is trivial since $r_1(t_0)=s_1(t_0)=r_1(1)=s_1(1)=1$.

	The sign of (\ref{eq:optvarcomp}) is therefore determined by the sign of $\pi(e^\epsilon+2\delta-1)-\delta(e^\epsilon+\delta)$, which gives $g(\epsilon, \delta)$ when solved for $\pi$. Hence, $\Var(\hat{\Pi}(r_\delta(t_0),s_\delta(t_0))| \pi) < \Var(\hat{\Pi}(r_\delta(1),s_\delta(1))| \pi)$ when $g(\epsilon,\delta) > \pi$. The other results follow similarly.
	\end{proof}

\textbf{Remark:} When $g(\epsilon,\delta) \le \pi$, the optimal mechanism corresponds with that established for $\epsilon$-differential privacy on RR (with an added dependence for $\delta$) and also with the optimal mechanism established in Theorem~10 of \cite{HLM16} for mechanisms on categorical data.  However, when $g(\epsilon,\delta) > \pi$, the optimal mechanism is one which we have not encountered previously.

The next corollary establishes the optimal mechanism(s) when $\pi \ge \frac{1}{2}$, and follows from Theorem~\ref{th:opteddp} by the symmetry of $\Var(\hat{\Pi}(p_{00},p_{11})| \pi)$ in $p_{00}$ and $p_{11}$.

\begin{corollary}\label{cr:opteddp}
Let $\delta > 0$ and $\frac{1}{2} \le \pi < 1$. Then, for $r_\delta$ and $s_\delta$ given by (\ref{eq:eddpparam}) and $g$ given by (\ref{eq:g_ed}),
$$\argmin_{(p_{00},p_{11}) \in \mathcal{R}^\prime} \Var(\hat{\Pi}(p_{00},p_{11})|\pi)
 = \begin{cases} \{(s_\delta(t_0),r_\delta(t_0))\}, & \text{if } g(\epsilon,\delta) > 1-\pi,\\
\{(s_\delta(1),r_\delta(1))\}, & \text{if } g(\epsilon,\delta) < 1-\pi,\\
\{(s_\delta(t_0),r_\delta(t_0)), (s_\delta(1),r_\delta(1))\}, & \text{if } g(\epsilon,\delta) = 1-\pi,\end{cases}$$
where $t_0 = t_0(\epsilon,\delta)$.
\end{corollary}

	\begin{proof}
	The result follows from Theorem~\ref{th:opteddp} since
	\begin{align*}
	\Var(\hat{\Pi}(p_{00}, p_{11})|\pi) &= \Var(\hat{\Pi}(p_{11}, p_{00})| 1-\pi).\qedhere
	\end{align*}
	\end{proof}

Example~\ref{eg:eddp} and Figure~\ref{fig:eddp} illustrate the conclusion of Theorem~\ref{th:opteddp}.

	\begin{example}\label{eg:eddp}
	Consider Theorem~\ref{th:opteddp} and Corollary~\ref{cr:opteddp} for various values of $\epsilon$, $\delta$ and $\pi$. For simplicity, in each of these examples we set $n=1$.
	
	\begin{enumerate}
	\item $\epsilon=\frac{1}{2}$, $\delta=\frac{1}{10}$, $\pi=\frac{1}{4}$: In this case, we have $g(\epsilon,\delta) = 0.243 < \pi$. Hence, the design matrix of the optimal mechanism is denoted by
	$$\left(\begin{array}{cc} \frac{e^\epsilon+\delta}{e^\epsilon+1} & \frac{1-\delta}{e^\epsilon+1}\\
	\frac{1-\delta}{e^\epsilon+1} & \frac{e^\epsilon+\delta}{e^\epsilon+1}\end{array}\right).$$
	
	This can be verified by noting that $\Var(\hat{\Pi}(r_\delta(1), s_\delta(1))| \pi) = 2.372$ and $\Var(\hat{\Pi}(r_\delta(t_0),s_\delta(t_0))| \pi) = 2.438$.
	\item $\epsilon=1$, $\delta=\frac{2}{5}$, $\pi=\frac{1}{10}$: In this case, $g(\epsilon,\delta)=0.197 > \pi$. Hence, the design matrix of the optimal mechanism is denoted by
	$$\left(\begin{array}{cc} 1 & 0 \\
	1-\delta & \delta \end{array}\right).$$
	
	Again, this can be verified by noting that $\Var(\hat{\Pi}(r_\delta(1), s_\delta(1))| \pi) = 0.385$ and $\Var(\hat{\Pi}(r_\delta(t_0),s_\delta(t_0))| \pi) = 0.24$.
	\item $\epsilon=\frac{1}{2}$, $\delta=\frac{1}{3}$, $\pi=\frac{9}{10}$: Since $\pi\ge \frac{1}{2}$, we use Corollary~\ref{cr:opteddp} for this example.  We note that $g(\epsilon,\delta) = 0.382 > 1-\pi$. Hence, the design matrix of the optimal mechanism is denoted by
	$$\left(\begin{array}{cc} \delta & 1-\delta \\
	0 & 1
	\end{array}\right).$$
	
	We see that $\Var(\hat{\Pi}(s_\delta(1), r_\delta(1))| \pi) = 0.854$ and $\Var(\hat{\Pi}(s_\delta(t_0), r_\delta(t_0))| \pi) = 0.143$. Note also that $\Var(\hat{\Pi}(r_\delta(0), s_\delta(0))| \pi) = 1.911$, corresponding with the conclusion of Lemma~\ref{lm:argminrssr}
	\item $\epsilon=\ln(2), \delta=\frac{1}{4}, \pi=\frac{1}{4}$: In this case, we have $g(\epsilon, \delta) = \frac{1}{4} = \pi$, hence there are two optimal mechanisms,
	$$\left(\begin{array}{cc} \frac{e^\epsilon+\delta}{e^\epsilon+1} & \frac{1-\delta}{e^\epsilon+1}\\
	\frac{1-\delta}{e^\epsilon+1} & \frac{e^\epsilon+\delta}{e^\epsilon+1}\end{array}\right), \left(\begin{array}{cc} 1 & 0 \\
	1-\delta & \delta \end{array}\right).$$
	
	This can be verified by noting that $\Var(\hat{\Pi}(r_\delta(1), s_\delta(1))| \pi)=\Var(\hat{\Pi}(r_\delta(t_0), s_\delta(t_0))| \pi)=\frac{15}{16}$.
	\end{enumerate}
	\end{example}

\begin{figure}
  \centering
\includegraphics[width=0.99\columnwidth]{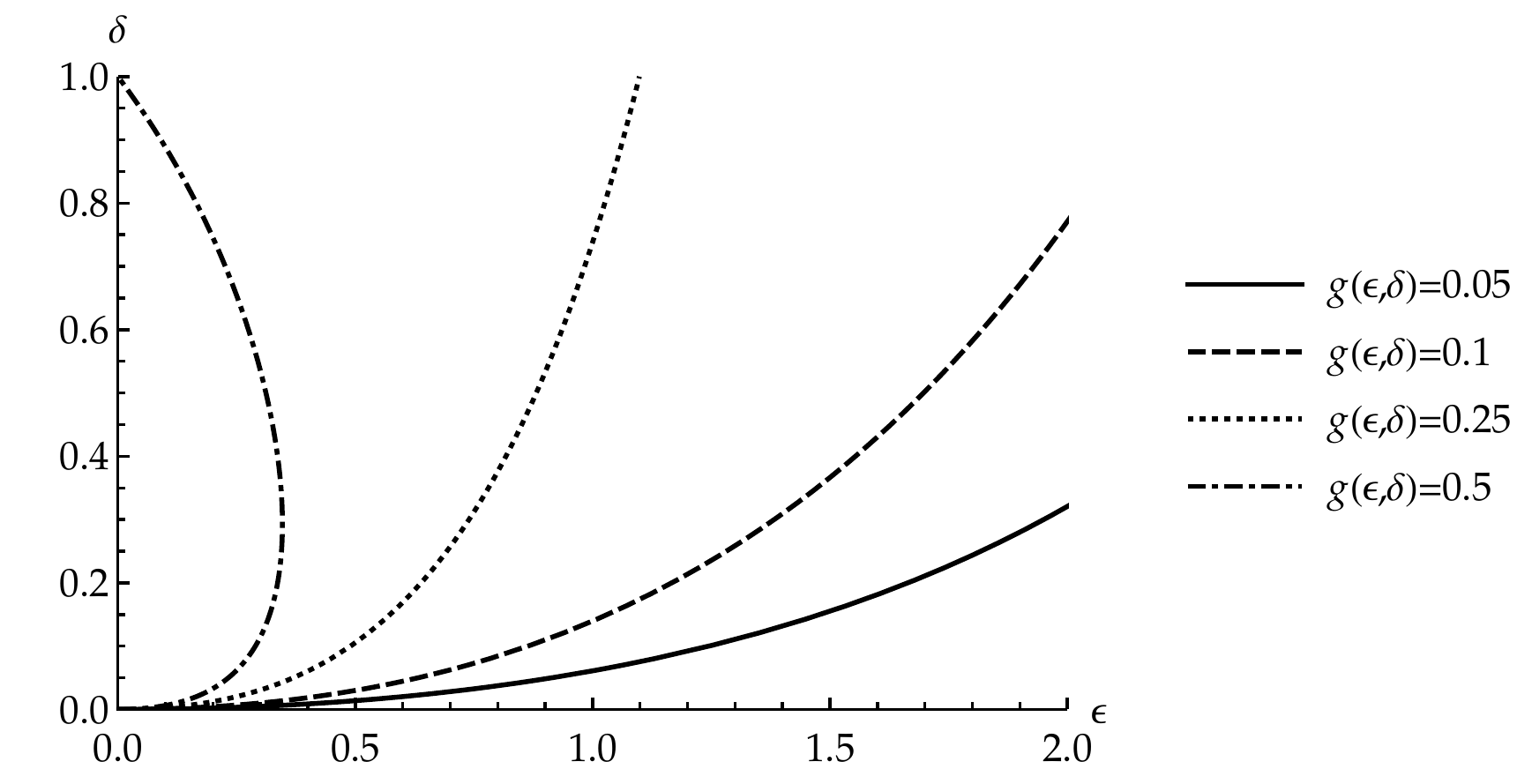}
  \caption[Contour plot of various level sets of $g(\epsilon,\delta)$]{A contour plot of various level sets of $g(\epsilon,\delta)$.  Given $\pi$, $\epsilon$ and $\delta$, these level sets can be used to determine the optimal {\ed}-differentially private RR mechanism.}
  \label{fig:eddp}
\end{figure}

\section{Optimal Warner Mechanism for \texorpdfstring{{\ed}-}{Relaxed }Differential Privacy}\label{sc:optwar}

In the final result of this paper, we examine the optimal mechanism for Warner's RR mechanism. We recall that Warner's mechanism imposed the additional constraint that $p_{00} = p_{11} = p_w$, so the design matrix becomes
$$\left(\begin{array}{cc} p_w & 1-p_w \\ 1-p_w & p_w \end{array}\right).$$
The error of such a mechanism is only a function of $p_w$ and the population proportion $\pi$, as shown in (\ref{eq:warnererror}).

As before, we require $2p_w > 1$. Our region of feasibility is therefore
$$\mathcal{R}_w = \left(\frac{1}{2}, \frac{e^\epsilon+\delta}{e^\epsilon+1}\right].$$

\begin{theorem}
Consider Warner's RR mechanism as presented in Section~\ref{sc:warner}. Then,
$$\argmin_{p_w \in \mathcal{R}_w} \Var(\hat{\Pi}_w(p_w)|\pi) = \left\{\frac{e^\epsilon+\delta}{e^\epsilon+1}\right\}.$$
\end{theorem}

	\begin{proof}
	By (\ref{eq:warnererror}), we note that
	$$\frac{\partial \Var(\hat{\Pi}_w(p_w)|\pi)}{\partial p_w} = \frac{1}{(1-2p_w)^3 n},$$
	hence $\frac{\partial \Var(\hat{\Pi}_w(p_w)|\pi)}{\partial p_w}<0$ when $p_w > \frac{1}{2}$. Therefore,
	$$\argmin_{p_w \in \mathcal{R}_w} \Var(\hat{\Pi}_w(p_w)|\pi) = \max (\mathcal{R}_w),$$
	and the result follows.
	\end{proof}

\section{Conclusions}\label{sc:conc}

We have presented the optimal differentially private RR mechanisms with respect to a maximum likelihood estimator, where both strict and relaxed differential privacy were considered.  For a given desired level of privacy, as determined by $\epsilon$ and $\delta$, we presented a method to quickly determine the optimal mechanism.  This will allow for the optimal implementation of differential privacy in any randomised response survey.

\section*{Acknowledgement}

The first named author was supported by the Science Foundation Ireland grant SFI/11/PI/1177.


\end{document}